\documentclass[11pt]{article}

\usepackage{amssymb,amsmath,amsthm,authblk,geometry,caption,natbib,mathrsfs,titlesec,tikz,subcaption,placeins,indentfirst,hyperref}
\usetikzlibrary{arrows.meta}

\definecolor{niceblue}{HTML}{0072B2}

\hypersetup{
  colorlinks=true,
  citecolor=niceblue,
  linkcolor=niceblue,
  urlcolor=niceblue,
  pdfborder={0 0 0}
}

\newtheoremstyle{myplain}%
  {1.0\baselineskip}
  {1.0\baselineskip}
  {\itshape}
  {}
  {\bfseries}
  {}
  { }
  {}

\newtheoremstyle{mydefinition}%
  {1.0\baselineskip}
  {1.0\baselineskip}
  {}
  {}
  {\bfseries}
  {}
  { }
  {}

\theoremstyle{myplain}
\newtheorem{theorem}{Theorem}
\newtheorem{proposition}{Proposition}
\newtheorem{lemma}{Lemma}
\newtheorem{corollary}{Corollary}

\theoremstyle{mydefinition}
\newtheorem{definition}{Definition}
\newtheorem{example}{Example}
\newtheorem{assumption}{Assumption}
\newtheorem{remark}{Remark}

\titlespacing*{\section}{0pt}{0.4\baselineskip}{0.6\baselineskip}
\titlespacing*{\subsection}{0pt}{0.6\baselineskip}{0.6\baselineskip}


\tikzset{
    basic/.style={
      node distance=2cm,
      every node/.style={circle, minimum size=1cm, text=black},
      >={Stealth},
      thick,
      scale=0.88,
      transform shape,
      color=niceblue
    }
}

\title{Learning to Agree under Pseudo-Reciprocity}
\author{Shinya Sugiura\thanks{\textit{E-mail address:} shinya.sugiura06@gmail.com}}
\affil{Department of Economics, Hitotsubashi University}

\date{August 2026}

\begin{document}

\maketitle

\begin{abstract}
\noindent
We characterize the communication networks that ensure consensus in rational social learning under the sure-thing principle. A network ensures consensus if and only if it is \textit{pseudo-reciprocal}, a property under which reciprocity may fail between individuals but holds at a coarser group level. Prior work has assumed message exchange to be bidirectional, yet such pairwise reciprocity proves largely dispensable. We show that a single bidirectional link can suffice, however large the population.


\end{abstract}

\bigskip

\newpage
\section{Introduction} \label{sec:introduction}
What an agent learns depends not only on whom she hears from, but also on whom those agents hear from, and so on. These communication links form a network, which determines whether repeated communication brings the agents to agreement. Which communication networks ensure consensus? A classical answer builds on the agreement theorem of \citet{aumann}. When agents with a common prior in a complete network repeatedly announce their posteriors and update on one another's announcements, they eventually reach a common posterior \citep{geanakoplosandpolemarchakis}. This conclusion has been extended to a decision-theoretic setting in which agents follow a common behavioral rule satisfying \citeauthor{savage}'s (\citeyear{savage}) \textit{sure-thing principle}, a fundamental rationality requirement.\footnote{See, e.g., \citet{krasucki}, \citet{muellerfrank}, and \citet{crescenzi}. Some earlier contributions impose conditions stronger than the sure-thing principle; see, e.g., \citet{parikhandkrasucki} and \citet{menager}.} Notably, later work \citep{muellerfrank,crescenzi} shows that consensus obtains for a much broader class of networks: it suffices that (i) information can flow from any agent to any other, possibly indirectly (\textit{strong connectedness}), and (ii) whenever one agent can send messages to another, the latter can send messages back (\textit{pairwise reciprocity}).

Despite this progress, the basic question has remained open: under the sure-thing principle, exactly which networks ensure consensus? We introduce a network property, \textit{pseudo-reciprocity}, and show that a network ensures consensus if and only if it is pseudo-reciprocal. Intuitively, pseudo-reciprocity captures situations in which reciprocity may fail between individuals but emerges at a coarser, group level. Consider a firm in which a manager delegates tasks through a team leader, while progress reports return through a separate project coordinator. The leader and the coordinator send messages to each other. Although no individual employee communicates with the manager in both directions, the manager and the team as a whole do: directives flow into the team through one channel, and feedback returns through another. Pseudo-reciprocity formalizes this group-level view.

Our main contribution is to identify the form of reciprocity that consensus requires. It is not the pairwise exchange assumed throughout the literature \citep{cave, krasucki, muellerfrank, crescenzi}, but reciprocity at the group level. This makes pairwise reciprocity largely dispensable: even a single bidirectional communication link can suffice, however large the population, and Corollary~\ref{cor:single-core-characterization} characterizes exactly when it does. A second contribution is technical, and has two parts. First, our results hold for state spaces of arbitrary cardinality and for communication of any ordinal length, a setting that only \citet{crescenzi} has previously admitted. Second, the characterization is operational: pseudo-reciprocity can be decided by an iterative procedure whose outcome does not depend on the order of its steps (Proposition~\ref{prop:confluence} and Corollary~\ref{cor:greedy}).

The paper proceeds as follows. Section~\ref{sec:preliminaries} presents the model. Section~\ref{sec:pseudo-reciprocity} defines pseudo-reciprocity. Section~\ref{sec:results} states the results, and Section~\ref{sec:proofs} collects the proofs.

\section{Preliminaries} \label{sec:preliminaries}
\subsection{Basic notions}
Let $\Omega$ be a set of states of the world. The communication structure is represented by a directed network $(I, g)$, where $I = \{1, \ldots, n\}$ with $n \geq 2$ is the set of agents and $g:I \times I \rightarrow \{0,1\}$ is the link function: $g(i,j)=1$ means that agent $i$ can send messages to agent $j$. We assume $g(i,i)=0$ for all $i \in I$. Each agent $i$ is endowed with a partition $\Pi_i$ of $\Omega$. For any $\omega \in \Omega$, we write $\Pi_i(\omega)$ for the element of $\Pi_i$ that contains $\omega$. Agent~$i$ observes $\Pi_i(\omega)$ as her private information when the true state is $\omega$.

Denote by $\boldsymbol\Pi$ the collection of all partitions of $\Omega$. We say that a partition $\Pi$ is \textit{finer} than a partition $\Pi^{\prime}$, and $\Pi^{\prime}$ \textit{coarser} than $\Pi$, if $\Pi(\omega) \subseteq \Pi^{\prime}(\omega)$ for all $\omega \in \Omega$; this is denoted by $\Pi^{\prime} \leq \Pi$. It is well known that the partially ordered set $(\boldsymbol{\Pi}, \leq)$ forms a complete lattice. For any pair $\Pi, \Pi^{\prime} \in \boldsymbol{\Pi}$, we denote their meet by $\Pi \wedge \Pi^{\prime}$ and their join by $\Pi \vee \Pi^{\prime}$.

Let $D$ be a non-empty set of decisions. Each agent $i$'s decision rule $f_i:\Omega\rightarrow D$ is determined by a common function $f:2^{\Omega}\rightarrow D$ with $f_i(\omega)=f(\Pi_i(\omega))$. This captures the assumption that the agents are \textit{like-minded}: they make identical decisions when given the same information. That is, for any $\omega \in \Omega$ and any $i,j \in I$, if $\Pi_{i}(\omega)=\Pi_{j}(\omega)$, then $f_{i}(\omega)=f_{j}(\omega)$. We call a triple $\mathscr{D}=(\Omega,D,f)$ a \textit{decision environment}. For a partition profile $\Pi=(\Pi_1,\ldots,\Pi_n)\in\boldsymbol{\Pi}^n$, we write $f_i^{\Pi}:=f\circ\Pi_i$ for agent $i$'s decision rule under $\Pi$. We define the function $W: \boldsymbol{\Pi} \rightarrow \boldsymbol{\Pi}$ by $W(\Pi)(\omega)=\{\omega^{\prime} \in \Omega \mid f(\Pi(\omega^{\prime}))=f(\Pi(\omega))\}$ for each $\omega \in \Omega$. For each $i \in I$, $W(\Pi_i)$ is called the \textit{working partition} of agent $i$. The cell $W(\Pi_i)(\omega)$ is the event in which agent $i$ chooses the decision $f_i(\omega)$. Note that, by definition, $W(\Pi_i)$ is coarser than $\Pi_i$.

We maintain two assumptions, one on decisions and one on the network; Section~\ref{sec:local} dispenses with the second. The first is due to \citet{savage} and was introduced into 
the consensus literature by \citet{cave} and \citet{bacharach}.

\begin{assumption}[Sure-thing principle]\label{ass:stp}
The decision function $f \colon 2^{\Omega} \to D$ satisfies the 
\textit{sure-thing principle} (STP): for any non-empty family $\mathcal{E} \subseteq 2^{\Omega}$ of pairwise disjoint non-empty events and any $d \in D$, 
if $f(E) = d$ for every $E \in \mathcal{E}$, then 
$f\!\left( \bigcup_{E \in \mathcal{E}} E \right) = d$.
\end{assumption}

STP has been viewed as a fundamental requirement of rationality (see, e.g., \citealp{bacharach}). For instance, if an investor would acquire a firm whether the incumbent or the challenger wins an upcoming election, then STP requires her to acquire it even before the result is known. 

The canonical formal example is conditional probability: for any probability measure $\mu$ on $\Omega$ and any event $F \subseteq \Omega$, the map $E \mapsto \mu(F \mid E)$ satisfies STP on events of positive $\mu$-measure. STP also covers non-probabilistic decisions. Consider an agent
choosing from a finite ordered set $D=\{d_1,\ldots,d_m\}$ with
$d_1\succ\cdots\succ d_m$. Higher-ranked decisions are more desirable but may
be infeasible in some states. For each state $\omega$, let
$C(\omega)\subseteq D$ be the non-empty set of feasible decisions, and assume
the fallback $d_m$ is feasible in every state. For an information event $E$,
define
\[
    C(E):=\bigcap_{\omega\in E} C(\omega),
    \qquad
    f_C(E):=\max_{\succ} C(E),
\]
with the convention $C(\emptyset)=D$. Then $f_C(E)$ is the most desirable
decision feasible in every state consistent with $E$, and $f_C$ satisfies STP.

The second assumption concerns the communication network. For a
network $(I,g)$, let $\tilde g(i,j):=\max\{g(i,j),g(j,i)\}$ be its underlying undirected relation.

\begin{assumption}[Weak connectedness]\label{ass:weak-connectedness}
The communication network $(I,g)$ is \textit{weakly connected}: for any $i,j \in I$ with $i \neq j$, there exists a path $i_0,\ldots,i_K\in I$ such that $i_0=i$, $i_K=j$, and $\tilde g(i_{k-1},i_k)=1$ for all $k \in \{1,\ldots, K\}$.
\end{assumption}

Weak connectedness is necessary for a network to ensure consensus: were it to fail, the agents would split into groups with no communication path---and hence no information flow---between them. Assuming it therefore involves no loss of generality; Section~\ref{sec:local} makes this precise.

\subsection{Communication and updating}
We now describe how information is revised through repeated message exchange.
Communication follows the fixed network $(I,g)$. At each stage, every
agent receives messages from her \textit{senders}: the set of senders of
agent $i$ is $S(i)=\{ j \in I \mid g(j,i)=1\}$.

At state $\omega$, each agent $j \in S(i)$ sends to $i$ a message conveying her decision $f_j(\omega)=f(\Pi_j(\omega))$. We assume that the partition profile, the decision function $f$, and the network $(I,g)$ are commonly known among the agents, while the realization of each agent's cell remains her private information. Agent $i$ can therefore compute the working partition $W(\Pi_j)$ of each sender $j\in S(i)$, and the message $f_j(\omega)$ reveals to her the event $W(\Pi_j)(\omega)$. 

Upon receiving these messages, agent $i$ updates her partition $\Pi_i$ using an \textit{information update function} $h_i:\boldsymbol{\Pi}^n \rightarrow \boldsymbol{\Pi}$, where $\boldsymbol{\Pi}^n$ denotes the $n$-fold product of $\boldsymbol{\Pi}$. We define
\begin{equation}
    h_i(\Pi_1,\ldots,\Pi_n)
    :=\Pi_i\vee \bigvee_{j\in S(i)}W(\Pi_j),
    \label{eq:update}
\end{equation}
where the join over the empty set is understood as the coarsest partition
$\{\Omega\}$. Hence, if $S(i)=\emptyset$, the update rule gives
$h_i(\Pi_1,\ldots,\Pi_n)=\Pi_i$.

In other words, once agent $i$ has received her senders' messages, she
refines her partition $\Pi_i$ to the coarsest common refinement of $\Pi_i$
and the senders' working partitions $\{W(\Pi_j)\}_{j\in S(i)}$. Thus, at state $\omega$, agent $i$ rules out all
states that are inconsistent with the event
\[
    \Pi_i(\omega)\cap\bigcap_{j\in S(i)}W(\Pi_j)(\omega).
\]
This updating procedure is standard in the literature
\citep{parikhandkrasucki, krasucki, houyandmenager, tsakasandvoorneveld}, except
that here communication follows a fixed network as in \citet{muellerfrank} and
\citet{crescenzi}.

Given an initial partition profile
$\Pi^0=(\Pi_1^0,\ldots,\Pi_n^0)\in\boldsymbol\Pi^n$, define the learning process
$(\Pi_i^\alpha)_{i\in I,\alpha}$ recursively over all ordinals as follows.
For every ordinal $\alpha$, write $\Pi^\alpha:=(\Pi_1^\alpha,\ldots,\Pi_n^\alpha)$.
At successor ordinals, define
\[
    \Pi_i^{\alpha+1}:=h_i(\Pi^\alpha)
    =\Pi_i^\alpha\vee\bigvee_{j\in S(i)}W(\Pi_j^\alpha)
    \qquad\text{for all }i\in I,
\]
and at every nonzero limit ordinal $\lambda$, define
\[
    \Pi_i^\lambda:=\bigvee_{\alpha<\lambda}\Pi_i^\alpha
    \qquad\text{for all }i\in I.
\]
Since the update rule \eqref{eq:update} is also commonly known, the profile
$\Pi^\alpha$ is commonly known at every stage $\alpha$, and the message sent by
agent $j$ at stage $\alpha$ reveals to its recipients the event
$W(\Pi_j^\alpha)(\omega)$.

\subsection{Consensus}

We define consensus through the steady states of the communication process. For a decision environment $\mathscr{D}=(\Omega,D,f)$, let 
\begin{align*}
    F_{\mathscr{D}}(I,g)
    &:=\bigl\{(\Pi_1,\ldots,\Pi_n)\in\boldsymbol{\Pi}^n:
        h_i(\Pi_1,\ldots,\Pi_n)=\Pi_i\text{ for all }i\in I\bigr\},\\
    A_{\mathscr{D}}
    &:=\bigl\{(\Pi_1,\ldots,\Pi_n)\in\boldsymbol{\Pi}^n:
        f(\Pi_i(\omega))=f(\Pi_j(\omega))
        \text{ for all }i,j\in I\text{ and }\omega\in\Omega\bigr\}.
\end{align*}

The set $F_{\mathscr D}(I,g)$ is the set of fixed points of the update rule: profiles
that no further round of communication refines. From any initial profile, the
process $(\Pi^\alpha)$ is increasing in $\alpha$. Fix an ordinal $\kappa$ whose
cardinality exceeds that of $\boldsymbol\Pi^n$. The profiles
$(\Pi^\alpha)_{\alpha<\kappa}$ cannot then be pairwise distinct, so
$\Pi^\alpha=\Pi^\beta$ for some $\alpha<\beta<\kappa$. Since the chain is
increasing, $\Pi^\alpha\leq\Pi^{\alpha+1}\leq\Pi^\beta=\Pi^\alpha$, whence
$\Pi^{\alpha+1}=\Pi^\alpha$. The chain therefore stabilizes, whether the state
space is finite or infinite. The stabilized profile satisfies
$h_i(\Pi^\alpha)=\Pi_i^\alpha$ for all $i\in I$ and hence lies in
$F_{\mathscr D}(I,g)$. Conversely, every profile in $F_{\mathscr D}(I,g)$ is
stationary. Thus $F_{\mathscr D}(I,g)$ is precisely the set of profiles at which
the process can settle. The set $A_{\mathscr D}$, in turn, consists of the
agreeing profiles. Requiring $F_{\mathscr D}(I,g)\subseteq A_{\mathscr D}$
therefore says that learning ends in agreement from every initial profile.

\begin{definition}\label{def:consensus-guarantee}
A network $(I,g)$ \emph{ensures consensus} if
$F_{\mathscr D}(I,g)\subseteq A_{\mathscr D}$ for every decision environment
$\mathscr D=(\Omega,D,f)$.
\end{definition}

\section{Pseudo-reciprocity}\label{sec:pseudo-reciprocity}

We now formalize pseudo-reciprocity. It is a property not of individual links but of the network viewed at a coarser level: agents are collected into groups, and reciprocity is evaluated between these groups. To make this idea precise, we first define networks whose nodes are groups of agents, and then introduce a merger operation that combines two groups whenever they are mutually linked.

A \emph{grouping} of $I$ is a partition $\mathcal P$ of $I$; its elements
are called \emph{blocks}. A
\emph{network on $\mathcal P$}, or a \emph{grouped network}, is a pair $(\mathcal P,g_{\mathcal P})$, where
$g_{\mathcal P}:\mathcal P\times\mathcal P\to\{0,1\}$ and
$g_{\mathcal P}(P,P)=0$ for every $P\in\mathcal P$. The original network
$(I,g)$ is identified with the network on the singleton grouping
$\mathcal I:=\{\{i\}:i\in I\}$, where
$g_{\mathcal I}(\{i\},\{j\})=g(i,j)$.

The pairs of groups we will merge are those linked in both directions. We
say that a two-element subset $\{P,Q\}\subseteq\mathcal P$ is a
\emph{mutual pair} of $(\mathcal P,g_{\mathcal P})$ if
$g_{\mathcal P}(P,Q)=g_{\mathcal P}(Q,P)=1$, and we write
$\mathscr M(\mathcal P,g_{\mathcal P})$ for the set of mutual pairs of
$(\mathcal P,g_{\mathcal P})$.
A mutual pair of the original network is called a \emph{direct mutual pair};
we identify a direct mutual pair $\{\{i\},\{j\}\}$ with the pair $\{i,j\}$ of
agents.
A network $(\mathcal P,g_{\mathcal P})$ is \emph{(pairwise) reciprocal} if
$g_{\mathcal P}(P,Q)=g_{\mathcal P}(Q,P)$ for all
$P,Q\in\mathcal P$; in particular, the original network is pairwise
reciprocal if $g(i,j)=g(j,i)$ for all $i,j\in I$. Pairwise reciprocity of
the original network, together with strong connectedness, is the weakest
sufficient condition for consensus established in prior work \citep{cave, krasucki, muellerfrank, crescenzi}.\footnote{A network $(I,g)$ is \emph{strongly connected} if, for any $i,j\in I$ with $i\neq j$, there exists a path $i_0,\ldots,i_K\in I$ such that $i_0=i$, $i_K=j$, and $g(i_{k-1},i_k)=1$ for all $k\in\{1,\ldots,K\}$.}

\begin{definition}[Merger]\label{def:merger}
Let $\{P,Q\}\in \mathscr M(\mathcal P,g_{\mathcal P})$. Define the quotient
map $\phi:\mathcal P\to \phi(\mathcal P)$ by
\[
\phi(X):=\begin{cases}
P\cup Q, & X\in\{P,Q\},\\
X, & \text{otherwise.}
\end{cases}
\]
The \emph{merger} of $P$ and $Q$ replaces $P$ and $Q$ with $P\cup Q$ and
produces the network $(\phi(\mathcal P),g_{\phi(\mathcal P)})$ defined by
\[
    g_{\phi(\mathcal P)}(R,S):=
    \begin{cases}
    \max\{g_{\mathcal P}(A,B):\phi(A)=R,\ \phi(B)=S\}, & R\neq S,\\
    0, & R=S.
    \end{cases}
\]
Thus the merged group inherits every incoming and outgoing link of its members,
while self-loops are suppressed.
\end{definition}

This operation is the analogue of edge contraction in graph minor theory, but merging is allowed only along mutual pairs. Unlike ordinary graph minors,
there is no deletion operation; mergers only coarsen the grouping.

\begin{definition}\label{def:mutual-minor}
A network $(\mathcal Q,g_{\mathcal Q})$ is a \emph{mutual minor} of
$(\mathcal P,g_{\mathcal P})$, written
$(\mathcal Q,g_{\mathcal Q})\preceq(\mathcal P,g_{\mathcal P})$, if it can be
obtained from $(\mathcal P,g_{\mathcal P})$ by a finite (possibly empty)
sequence of mergers.
\end{definition}

Intuitively, a mutual minor of $(\mathcal P,g_{\mathcal P})$ is a coarser view of it obtained by fusing pairs of groups between which messages flow in both directions. The relation $\preceq$ is a partial order. Pseudo-reciprocal networks are defined by this relation.

\begin{definition}\label{def:pseudo-reciprocity}
A network $(I,g)$ is \emph{pseudo-reciprocal} if it has a
reciprocal mutual minor.
\end{definition}

In other words, a network is pseudo-reciprocal if, after finitely many mergers of mutually linked groups, the resulting grouped network is reciprocal. If a network is already reciprocal, it is trivially pseudo-reciprocal, since it is its own reciprocal mutual minor. Pairwise reciprocity, the condition assumed throughout the prior literature, is therefore a special case of pseudo-reciprocity.

\begin{remark}\label{rem:strong-connectedness}
Under Assumption~\ref{ass:weak-connectedness}, every pseudo-reciprocal network is
strongly connected. This follows from
Lemmas~\ref{lem:block-connectivity} and \ref{lem:pseudo-one-block} in Section \ref{sec:proofs}.
\end{remark}

\begin{example}\label{ex:pseudo-reciprocity-stepwise}
Figure~\ref{fig:pseudo-reciprocity-stepwise} gives a merger sequence that
witnesses pseudo-reciprocity. Start from the singleton grouping
$\mathcal P_0=\{\{i\},\{j\},\{k\},\{\ell\},\{m\}\}$, shown in
Panel~\subref{fig:pseudo-step-a}. Write $g_t$ for the grouped link
function on $\mathcal P_t$. The only mutual pair of $(\mathcal P_0,g_0)$ is
$\{\{j\},\{k\}\}$, so the first merger necessarily replaces $\{j\}$ and
$\{k\}$ by $\{j,k\}$. This produces $\mathcal P_1=\{\{i\},\{j,k\},\{\ell\},\{m\}\}$, shown in Panel~\subref{fig:pseudo-step-b}.

By Definition~\ref{def:merger}, the block $\{j,k\}$ inherits all incoming and
outgoing links of $j$ and $k$. Hence the links $j\to i$ and $i\to k$ make
$\{j,k\}$ and $\{i\}$ mutually linked, while $k\to \ell$ and $\ell\to j$ do
the same for $\{j,k\}$ and $\{\ell\}$. In particular,
$\{\{j,k\},\{\ell\}\}$ is a mutual pair of $(\mathcal P_1,g_1)$.
Merging these two blocks gives $\mathcal P_2=\{\{i\},\{j,k,\ell\},\{m\}\}$, as in Panel~\subref{fig:pseudo-step-c}.

The network on $\mathcal P_2$ is reciprocal: $i$ and $m$ are each
mutually linked with $\{j,k,\ell\}$, and there is no link in either direction
between $i$ and $m$. Therefore $(\mathcal P_2,g_2)$ is a reciprocal
mutual minor of $(\mathcal P_0,g_0)$. By
Definition~\ref{def:pseudo-reciprocity}, the original network is
pseudo-reciprocal.
\end{example}

\begin{figure}[!htb]
\centering
\begin{subfigure}[b]{0.41\textwidth}
\centering
\begin{tikzpicture}[basic,scale=1.06,transform shape]
    \node (i) at (0,0) {$i$};
    \node (j) at (2.5,1) {$j$};
    \node (ell) at (5,1) {$\ell$};
    \node (k) at (2.5,-1) {$k$};
    \node (m) at (5,-1) {$m$};
    \draw[->] (j) -- (i);
    \draw[->] (k) -- (ell);
    \draw[<->] (k) -- (j);
    \draw[->] (ell) -- (j);
    \draw[->] (i) -- (k);
    \draw[->] (ell) -- (m);
    \draw[->] (m) -- (k);
\end{tikzpicture}
\caption{Original network}
\label{fig:pseudo-step-a}
\end{subfigure}
\begin{subfigure}[b]{0.41\textwidth}
\centering
\begin{tikzpicture}[basic,scale=1.06,transform shape]
    \path (0,1) -- (0,-1.5);
    \node (mi) at (0,0) {$i$};
    \node (mjk) at (2.5,0) {$\{j,k\}$};
    \node (mell) at (5,1) {$\ell$};
    \node (mm) at (5,-1) {$m$};
    \draw[<->] (mi) -- (mjk);
    \draw[<->] (mjk) -- (mell);
    \draw[->] (mell) -- (mm);
    \draw[->] (mm) -- (mjk);
\end{tikzpicture}
\caption{After merging $\{j\}$ and $\{k\}$}
\label{fig:pseudo-step-b}
\end{subfigure}

\vspace{0.6em}

\begin{subfigure}[b]{0.41\textwidth}
\centering
\begin{tikzpicture}[basic,scale=1.08,transform shape]
    \path (0,1) -- (0,-1.5);
    \node (i) at (0,0) {$i$};
    \node (jkell) at (2.5,0) {$\{j,k,\ell\}$};
    \node (m) at (5,0) {$m$};
    \draw[<->] (i) -- (jkell);
    \draw[<->] (m) -- (jkell);
\end{tikzpicture}
\caption{A reciprocal minor}
\label{fig:pseudo-step-c}
\end{subfigure}
\caption{A merger sequence witnessing pseudo-reciprocity}
\label{fig:pseudo-reciprocity-stepwise}
\end{figure}
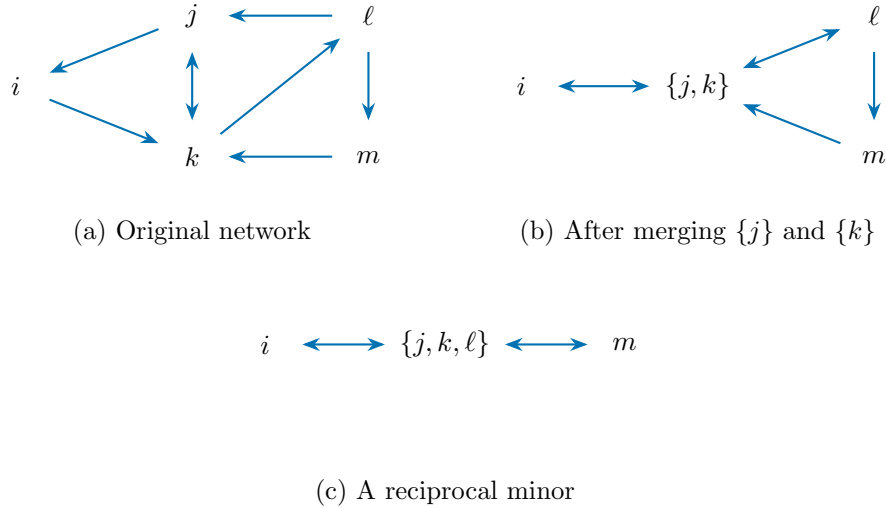

\section{Results}\label{sec:results}

\subsection{The characterization}\label{sec:characterization}

We now state the main result. In a pseudo-reciprocal network, communication
may be one-way between individuals; yet once mutually linked groups are
successively merged, every remaining link between groups is reciprocated. For
agents who obey the sure-thing principle, this group-level
reciprocity is exactly the condition under which consensus is ensured.

\begin{theorem}\label{thm:pseudo-reciprocity}
A network ensures consensus if and only if it is pseudo-reciprocal.
\end{theorem}

\begin{remark}\label{rem:finite-necessity}
The characterization holds for state spaces of arbitrary cardinality and
communication of any ordinal length, generalizing \citet{crescenzi}, the only
prior work to admit this setting. Moreover, it is unchanged if
Definition~\ref{def:consensus-guarantee} is restricted to finite decision
environments: sufficiency holds for arbitrary state spaces, and necessity is
witnessed by a finite environment constructed in Section~\ref{sec:proofs-separation}.
\end{remark}

Definition~\ref{def:pseudo-reciprocity} requires only that some sequence of
mergers end in a reciprocal minor; the next result shows that the order of
mergers is in fact irrelevant.

Call a sequence of mergers \emph{maximal} if its terminal grouped network has
no mutual pair. Since every merger reduces the number of blocks by one, every
sequence of mergers is finite and extends to a maximal one.

\begin{proposition}\label{prop:confluence}
All maximal sequences of mergers from $(\mathcal I,g_{\mathcal I})$ terminate
at the same grouped network.
\end{proposition}

We call this common endpoint the \emph{terminal minor} of $(I,g)$.

\begin{corollary}\label{cor:greedy}
The following are equivalent:
\begin{enumerate}
    \item[\textnormal{(i)}] $(I,g)$ is pseudo-reciprocal.
    \item[\textnormal{(ii)}] Some maximal sequence of mergers from $(\mathcal I,g_{\mathcal I})$ terminates at the one-block grouping $\{I\}$.
    \item[\textnormal{(iii)}] Every maximal sequence of mergers from $(\mathcal I,g_{\mathcal I})$ terminates at $\{I\}$.
\end{enumerate}
\end{corollary}

Hence a network is pseudo-reciprocal if and only if its terminal minor is the
one-block grouping $\{I\}$. To decide pseudo-reciprocity, it suffices to merge
mutual pairs in any order until none remains; the procedure terminates within
$n-1$ steps.

\subsection{Consensus with a unique mutual pair}\label{sec:unique-mutual-pair}

By Corollary~\ref{cor:greedy}, a network is pseudo-reciprocal exactly when
the whole society can be assembled from individual agents by repeatedly
uniting two groups linked in both directions. Uniting two groups requires
only one channel in each direction. We now ask how few direct mutual pairs,
and how few links, consensus requires.

\begin{definition}\label{def:two-sided-group-attachment}
For disjoint non-empty groups $A,B\subseteq I$, define
$\bar g(A,B):=\max\{g(a,b):a\in A,\ b\in B\}$. We say that $A$ and $B$ are
\emph{two-sidedly attached}, written $A\rightleftarrows_g B$, if
$\bar g(A,B)=\bar g(B,A)=1$.
\end{definition}

A merged group inherits every link of its members
(Definition~\ref{def:merger}), so two groups can be merged exactly when they
are two-sidedly attached. Call a non-singleton block arising during the
assembly a \emph{core}. A new core arises only by uniting two singletons, and
two singletons can be merged only if they form a direct mutual pair. Every
core therefore starts from a direct mutual pair. Several cores may grow in
parallel and be joined later. But a network with a single direct mutual pair
can grow only one core, and every subsequent merger attaches one agent to it.
Consensus is ensured exactly when this core absorbs every other agent one at
a time.

\begin{definition}[Single-core order]\label{def:single-core-order}
A network $(I,g)$ \emph{admits a single-core order} if there is an ordering
$v_1,\ldots,v_n$ of $I$ such that, writing $M_k:=\{v_1,\ldots,v_k\}$,
\begin{equation}\label{eq:single-core-order}
    \mathscr M(I,g)=\{\{v_1,v_2\}\}
    \quad\text{and}\quad
    M_k\rightleftarrows_g \{v_{k+1}\}
    \quad\text{for every } k=2,\ldots,n-1.
\end{equation}
Such an ordering is a \emph{single-core order}, and the core
$M_2=\{v_1,v_2\}$ is its \emph{initial core}.
\end{definition}

Because the direct mutual pair is unique, the two channels that attach
$v_{k+1}$ to $M_k$ in \eqref{eq:single-core-order} cannot be the two
directions of a single dyad: otherwise $v_{k+1}$ and some incumbent in $M_k$
would form a second direct mutual pair.

\begin{corollary}\label{cor:single-core-characterization}
The following are equivalent for a network $(I,g)$:
\begin{enumerate}
    \item[\textnormal{(i)}] $(I,g)$ has exactly one direct mutual pair and ensures consensus.
    \item[\textnormal{(ii)}] $(I,g)$ has exactly one direct mutual pair and is pseudo-reciprocal.
    \item[\textnormal{(iii)}] $(I,g)$ admits a single-core order.
\end{enumerate}
\end{corollary}

Corollary~\ref{cor:single-core-characterization} isolates the minimal
reciprocity that consensus requires: a single reciprocal dyad, the initial
core, suffices for a population of any size, with every other link running
one way.

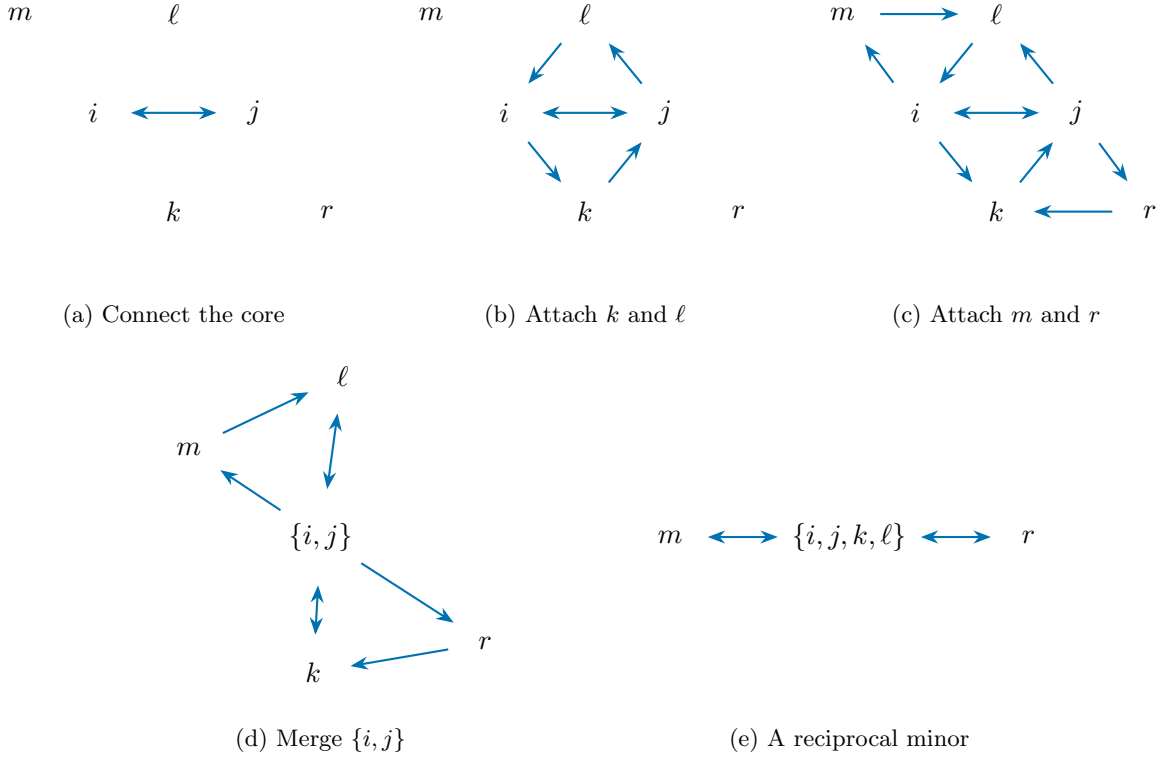
\begin{figure}[tb]
  \centering
  \captionsetup{font=small}
  \begin{subfigure}[b]{0.32\textwidth}
    \centering
    \begin{tikzpicture}[basic, scale=1.1, transform shape]
      \path[use as bounding box] (-2.55,-2.25) rectangle (2.55,2.25);
      \node (ni) at (-1.1,0)      {$i$};
      \node (nj) at (1.1,0)       {$j$};
      \node (nk) at (0,-1.35)     {$k$};
      \node (nl) at (0,1.35)      {$\ell$};
      \node (nm) at (-2.10,1.35)  {$m$};
      \node (nr) at (2.10,-1.35)  {$r$};
      \draw[<->] (ni) -- (nj);
    \end{tikzpicture}
    \caption{Connect the core}
    \label{fig:rc-construct-core}
  \end{subfigure}%
  \hfill
  \begin{subfigure}[b]{0.32\textwidth}
    \centering
    \begin{tikzpicture}[basic, scale=1.1, transform shape]
      \path[use as bounding box] (-2.55,-2.25) rectangle (2.55,2.25);
      \node (ni) at (-1.1,0)      {$i$};
      \node (nj) at (1.1,0)       {$j$};
      \node (nk) at (0,-1.35)     {$k$};
      \node (nl) at (0,1.35)      {$\ell$};
      \node (nm) at (-2.10,1.35)  {$m$};
      \node (nr) at (2.10,-1.35)  {$r$};
      \draw[<->] (ni) -- (nj);
      \draw[->]  (ni) -- (nk);
      \draw[->]  (nk) -- (nj);
      \draw[->]  (nj) -- (nl);
      \draw[->]  (nl) -- (ni);
    \end{tikzpicture}
    \caption{Attach $k$ and $\ell$}
    \label{fig:rc-construct-kl}
  \end{subfigure}%
  \hfill
  \begin{subfigure}[b]{0.32\textwidth}
    \centering
    \begin{tikzpicture}[basic, scale=1.1, transform shape]
      \path[use as bounding box] (-2.55,-2.25) rectangle (2.55,2.25);
      \node (ni) at (-1.1,0)      {$i$};
      \node (nj) at (1.1,0)       {$j$};
      \node (nk) at (0,-1.35)     {$k$};
      \node (nl) at (0,1.35)      {$\ell$};
      \node (nm) at (-2.10,1.35)  {$m$};
      \node (nr) at (2.10,-1.35)  {$r$};
      \draw[<->] (ni) -- (nj);
      \draw[->]  (ni) -- (nk);
      \draw[->]  (nk) -- (nj);
      \draw[->]  (nj) -- (nl);
      \draw[->]  (nl) -- (ni);
      \draw[->]  (ni) -- (nm);
      \draw[->]  (nm) -- (nl);
      \draw[->]  (nj) -- (nr);
      \draw[->]  (nr) -- (nk);
    \end{tikzpicture}
    \caption{Attach $m$ and $r$}
    \label{fig:rc-construct-final}
  \end{subfigure}

  \vspace{0.7\baselineskip}

  \begin{subfigure}[b]{0.42\textwidth}
    \centering
    \begin{tikzpicture}[basic, scale=1.1, transform shape]
      \path[use as bounding box] (-2.55,-2.25) rectangle (2.55,2.25);
      \node (gij) at (0,0)         {$\{i,j\}$};
      \node (bk) at (-0.10,-1.85) {$k$};
      \node (bl) at (0.30,2.20)   {$\ell$};
      \node (bm) at (-1.80,1.20)  {$m$};
      \node (br) at (2.25,-1.45)  {$r$};
      \draw[<->] (gij) -- (bl);
      \draw[<->] (gij) -- (bk);
      \draw[->]  (gij) -- (bm);
      \draw[->]  (gij) -- (br);
      \draw[->]  (bm) -- (bl);
      \draw[->]  (br) -- (bk);
    \end{tikzpicture}
    \caption{Merge $\{i,j\}$}
    \label{fig:rc-minor-first}
  \end{subfigure}
  \hspace{0.001\textwidth}
  \begin{subfigure}[b]{0.42\textwidth}
    \centering
    \begin{tikzpicture}[basic, scale=1.1, transform shape]
      \path[use as bounding box] (-2.80,-2.25) rectangle (2.80,2.25);
      \node (cm) at (-2.45,0)  {$m$};
      \node (gg) at (0,0)      {$\{i,j,k,\ell\}$};
      \node (cr) at (2.45,0)   {$r$};
      \draw[<->] (cm) -- (gg);
      \draw[<->] (gg) -- (cr);
    \end{tikzpicture}
    \caption{A reciprocal minor}
    \label{fig:rc-minor-final}
  \end{subfigure}
  \caption{Building a pseudo-reciprocal network from the single mutual pair $\{i,j\}$. Panels (a)--(c) attach each remaining agent through one channel in and one out; Panels (d)--(e) contract the core to a reciprocal minor.}
  \label{fig:reciprocal-core}
\end{figure}

Figure~\ref{fig:reciprocal-core} illustrates
Corollary~\ref{cor:single-core-characterization} with six agents. The pair
$\{i,j\}$ is the only direct mutual pair, and
Panels~(\subref{fig:rc-construct-core})--(\subref{fig:rc-construct-final})
build the network along a single-core order: each remaining agent is
attached to the growing core through one incoming and one outgoing channel,
as \eqref{eq:single-core-order} requires.
Panels~(\subref{fig:rc-minor-first})--(\subref{fig:rc-minor-final}) then
display the merger sequence that witnesses pseudo-reciprocity: merging
$\{i,j\}$ and absorbing $k$ and $\ell$ turns the one-way links of $m$ and
$r$ into mutual ones, leaving a reciprocal minor.

Reciprocity can thus be reduced to a single dyad; the following proposition
shows that the total number of links cannot be reduced below $2(n-1)$.

\begin{proposition}\label{prop:link-count}
If a network ensures consensus, then $\sum_{i,j\in I}g(i,j)\geq 2(n-1)$.
For every $n\geq 2$, equality holds for some network with $n$ agents that
admits a single-core order.
\end{proposition}

\begin{remark}\label{rem:link-count}
Strong connectedness alone requires only $n$ links, attained by a directed
cycle; for $n\geq 3$, the cycle has no mutual pair and hence does not ensure
consensus (Theorem~\ref{thm:pseudo-reciprocity} and
Corollary~\ref{cor:greedy}). Under pairwise reciprocity, weak connectedness
already forces at least $n-1$ mutual dyads and hence at least $2(n-1)$ links.
Pseudo-reciprocity therefore leaves the minimal number of links unchanged;
what it reduces is the number of mutual dyads required, from $n-1$ to one.
\end{remark}

\subsection{Local consensus}\label{sec:local}

Even without Assumption~\ref{ass:weak-connectedness}, the same analysis
characterizes consensus in a local form: we determine between which
agents consensus is ensured, and the blocks of the terminal minor emerge
as the maximal groups within which it is.

For a non-empty $J\subseteq I$, we write $g_J$ for the restriction of $g$
to $J\times J$ and call the pair $(J,g_J)$ the subnetwork \emph{induced}
by $J$. Weak connectedness and pseudo-reciprocity apply to $(J,g_J)$ as to
any network; both hold trivially when $|J|=1$. For $i,j\in I$, we say
that $(I,g)$ \emph{ensures consensus between $i$ and $j$} if
$f_i^\Pi=f_j^\Pi$ for every decision environment
$\mathscr D=(\Omega,D,f)$ and every $\Pi\in F_{\mathscr D}(I,g)$.

\begin{theorem}\label{thm:decomposition}
Let $i,j\in I$. The following are equivalent:
\begin{enumerate}
    \item[\textnormal{(i)}] The network ensures consensus between $i$ and $j$.
    \item[\textnormal{(ii)}] Some $J\subseteq I$ containing $i$ and $j$ induces a weakly connected, pseudo-reciprocal subnetwork.
    \item[\textnormal{(iii)}] $i$ and $j$ belong to the same block of the terminal minor.
\end{enumerate}
Moreover, the blocks of the terminal minor are precisely the maximal sets
inducing weakly connected, pseudo-reciprocal subnetworks.
\end{theorem}

Condition (ii) states that some group containing both agents would, as a
network in its own right, ensure consensus by
Theorem~\ref{thm:pseudo-reciprocity}. Condition
(iii) makes the criterion operational, since the
terminal minor is computed by merging mutual pairs in arbitrary order
(Proposition~\ref{prop:confluence}). The theorem is informative under
Assumption~\ref{ass:weak-connectedness} as well: when pseudo-reciprocity
fails, it identifies the pairs between which consensus is still ensured.

\begin{example}\label{ex:local-consensus}
Let $I=\{1,2,3,4\}$ with links $1\leftrightarrow 2$, $2\to 3$, and
$3\leftrightarrow 4$ only. Merging the two dyads yields the terminal minor:
the blocks $\{1,2\}$ and $\{3,4\}$, joined by a single link from the former
to the latter. By Theorem~\ref{thm:decomposition}, the network ensures
consensus within each dyad but not across them---even though the network is
weakly connected and information flows from $\{1,2\}$ to $\{3,4\}$. The
dyads are the maximal sets inducing weakly connected, pseudo-reciprocal
subnetworks.
\end{example}

Deleting the link $2\to 3$ makes the network pseudo-reciprocal---merging
the two dyads now leaves two blocks with no link between them, a vacuously
reciprocal minor---yet consensus between the dyads still fails. Weak
connectedness in condition (ii) therefore cannot be dropped.

Finally, the necessity claim of Section~\ref{sec:preliminaries} follows:
if $(I,g)$ is not weakly connected, there exist agents $i$ and $j$
with no path between them in the underlying undirected relation; no $J$
containing both then induces a weakly connected subnetwork, so, by
Theorem~\ref{thm:decomposition}, the network does not ensure consensus
between $i$ and $j$, and hence does not ensure consensus.

\FloatBarrier

\section{Proofs}\label{sec:proofs}

This section proves the results of Section~\ref{sec:results}. The
sufficiency part of Theorem~\ref{thm:pseudo-reciprocity} follows from
Lemmas~\ref{lem:pseudo-one-block} and \ref{lem:block-consensus}, and the
necessity part from Corollary~\ref{cor:lift}. The proof of
Theorem~\ref{thm:decomposition} uses
Lemmas~\ref{lem:blocks-pseudo-reciprocal} and \ref{lem:block-consensus}
together with Corollary~\ref{cor:lift}. The remaining results are derived
from Theorem~\ref{thm:pseudo-reciprocity} and the analysis of mergers in
Section~\ref{sec:proofs-mergers}.

\subsection{Mergers}\label{sec:proofs-mergers}

\begin{lemma}\label{lem:closed-form-mergers}
Let $(\mathcal P_m,g_{\mathcal P_m})$ be obtained from a grouped network
$(\mathcal P_0,g_{\mathcal P_0})$ by a sequence of $m$ mergers, and let
$\Phi_m:\mathcal P_0\to\mathcal P_m$ be the induced quotient map.
\begin{enumerate}
    \item[\textnormal{(i)}] For all distinct $X,Y\in\mathcal P_m$,
    \begin{equation}\label{eq:closed-form-results}
        g_{\mathcal P_m}(X,Y)
        =\max\{g_{\mathcal P_0}(A,B):\Phi_m(A)=X,\ \Phi_m(B)=Y\}.
    \end{equation}
    \item[\textnormal{(ii)}] If $(\mathcal P_0,g_{\mathcal P_0})$ is reciprocal,
    then so is $(\mathcal P_m,g_{\mathcal P_m})$.
    \item[\textnormal{(iii)}] Suppose $\mathcal P_0=\mathcal I$. Then
    $g_{\mathcal P_m}(X,Y)=\bar g(X,Y)$ for all distinct
    $X,Y\in\mathcal P_m$. Hence $(\mathcal P_m,g_{\mathcal P_m})$ is determined
    by its grouping, and $\{X,Y\}$ is a mutual pair if and only if
    $X\rightleftarrows_g Y$.
\end{enumerate}
\end{lemma}

\begin{proof}
Let
$(\mathcal P_0,g_{\mathcal P_0})\to(\mathcal P_1,g_{\mathcal P_1})\to\cdots\to(\mathcal P_m,g_{\mathcal P_m})$
be a merger sequence. Let $\phi_t:\mathcal P_{t-1}\to\mathcal P_t$ be the
quotient map at step $t$, and let $\Phi_t:=\phi_t\circ\cdots\circ\phi_1$.

\smallskip
\noindent(i) We induct on $m$. The case $m=0$ is immediate. Suppose the
formula holds after $m-1$ mergers. For distinct $X,Y\in\mathcal P_m$,
Definition~\ref{def:merger} gives
\begin{align*}
    g_{\mathcal P_m}(X,Y)
    &=\max_{\substack{\phi_m(X^{\prime})=X\\ \phi_m(Y^{\prime})=Y}}
        g_{\mathcal P_{m-1}}(X^{\prime},Y^{\prime})  \\
    &=\max_{\substack{\phi_m(X^{\prime})=X\\ \phi_m(Y^{\prime})=Y}}
        \max_{\substack{\Phi_{m-1}(A)=X^{\prime}\\ \Phi_{m-1}(B)=Y^{\prime}}}
        g_{\mathcal P_0}(A,B)  \\
    &=\max_{\substack{\Phi_m(A)=X\\ \Phi_m(B)=Y}}
        g_{\mathcal P_0}(A,B).
\end{align*}
The second equality is the induction hypothesis, which applies because
$X\neq Y$ forces $X^{\prime}\neq Y^{\prime}$. The last equality merely
rewrites the indexing set, since $\Phi_m=\phi_m\circ\Phi_{m-1}$.

\smallskip
\noindent(ii) Suppose $(\mathcal P_0,g_{\mathcal P_0})$ is reciprocal. For
distinct $X,Y\in\mathcal P_m$, part (i) gives
\begin{align*}
    g_{\mathcal P_m}(X,Y)
    &=\max_{\Phi_m(A)=X,\ \Phi_m(B)=Y}g_{\mathcal P_0}(A,B) \\
    &=\max_{\Phi_m(A)=X,\ \Phi_m(B)=Y}g_{\mathcal P_0}(B,A) \\
    &=\max_{\Phi_m(A)=Y,\ \Phi_m(B)=X}g_{\mathcal P_0}(A,B)
      =g_{\mathcal P_m}(Y,X).
\end{align*}
The case $X=Y$ follows from the no-self-loop convention.

\smallskip
\noindent(iii) Let $\mathcal P_0=\mathcal I$. Each quotient map sends a block
to the block containing it, so $\Phi_m(\{a\})$ is the block of
$\mathcal P_m$ containing $a$. Hence \eqref{eq:closed-form-results} reads
$g_{\mathcal P_m}(X,Y)=\max\{g(a,b):a\in X,\ b\in Y\}=\bar g(X,Y)$ for
distinct $X,Y\in\mathcal P_m$. The link function is thus determined by the
grouping, and $\{X,Y\}$ is a mutual pair if and only if
$\bar g(X,Y)=\bar g(Y,X)=1$, that is, $X\rightleftarrows_g Y$.
\end{proof}

The next lemma underlies Remark~\ref{rem:strong-connectedness} and is used
in the proof of Lemma~\ref{lem:blocks-pseudo-reciprocal}.

\begin{lemma}\label{lem:block-connectivity}
Let $(\mathcal P,g_{\mathcal P})$ be obtained from $(\mathcal I,g_{\mathcal I})$
by a sequence of mergers. Then, for every block $P\in\mathcal P$ and all
distinct $a,b\in P$, there exists a path $i_0,\ldots,i_K\in P$ such that
$i_0=a$, $i_K=b$, and $g(i_{k-1},i_k)=1$ for all $k\in\{1,\ldots,K\}$.
\end{lemma}

\begin{proof}
We induct on the length of the merger sequence. Singleton blocks satisfy the
claim vacuously. Suppose the claim holds at $\mathcal P$, and let
$\mathcal P^{\prime}$ be obtained from $\mathcal P$ by merging a mutual pair
$\{P,Q\}$. Only the new block $P\cup Q$ needs to be checked.

Since $g_{\mathcal P}(P,Q)=g_{\mathcal P}(Q,P)=1$,
Lemma~\ref{lem:closed-form-mergers}(iii) yields $a_1\in P$ and $b_1\in Q$ with
$g(a_1,b_1)=1$, and $a_2\in Q$ and $b_2\in P$ with $g(a_2,b_2)=1$. Let
$a,b\in P\cup Q$ be distinct. If $a,b\in P$ or $a,b\in Q$, the claim follows
from the induction hypothesis. If $a\in P$ and $b\in Q$, concatenating a path
from $a$ to $a_1$ within $P$ (empty if $a=a_1$), the link $a_1\to b_1$, and a
path from $b_1$ to $b$ within $Q$ (empty if $b_1=b$) yields the required path.
The case $a\in Q$ and $b\in P$ is symmetric, via the link $a_2\to b_2$.
\end{proof}

\begin{lemma}\label{lem:pseudo-one-block}
\textnormal{(i)} If the one-block grouping $\{I\}$ can be obtained from
$(\mathcal I,g_{\mathcal I})$ by a finite sequence of mergers, then $(I,g)$
is pseudo-reciprocal.
\textnormal{(ii)} If $(I,g)$ is weakly connected and pseudo-reciprocal, then
$\{I\}$ can be so obtained.
\end{lemma}

\begin{proof}
(i) The grouped network on $\{I\}$ is reciprocal. Hence, if $\{I\}$ can be
obtained from $(\mathcal I,g_{\mathcal I})$ by mergers, then $(I,g)$ has a
reciprocal mutual minor and is therefore pseudo-reciprocal.

(ii) Suppose that $(I,g)$ is weakly connected and pseudo-reciprocal. Choose a
reciprocal mutual minor
$(\mathcal Q,g_{\mathcal Q})\preceq(\mathcal I,g_{\mathcal I})$. The grouped
network $(\mathcal Q,g_{\mathcal Q})$ is weakly connected, since it is
obtained by quotienting the weakly connected network $(I,g)$. If
$|\mathcal Q|=1$, there is nothing to prove. If $|\mathcal Q|>1$, weak
connectedness gives distinct blocks $R,S\in\mathcal Q$ linked in at least one
direction. Since $(\mathcal Q,g_{\mathcal Q})$ is reciprocal, they are linked
in both:
\[
    g_{\mathcal Q}(R,S)=g_{\mathcal Q}(S,R)=1.
\]
Thus $R$ and $S$ can be merged. Merging preserves reciprocity by
Lemma~\ref{lem:closed-form-mergers}(ii), and quotienting preserves weak
connectedness. Repeating this operation finitely many times yields the
one-block grouping $\{I\}$.
\end{proof}

\begin{proof}[Proof of Proposition~\ref{prop:confluence}]
Call a grouping \emph{reachable} if a network on it can be obtained from
$(\mathcal I,g_{\mathcal I})$ by a finite (possibly empty) sequence of
mergers. By Lemma~\ref{lem:closed-form-mergers}(iii), such a network is
determined by its grouping, and its mutual pairs are exactly the pairs of
two-sidedly attached blocks. We therefore refer to reachable networks by
their groupings. The groupings obtained from a reachable grouping by a single
merger are its \emph{one-step successors}.

\smallskip
\noindent\emph{Step 1: any two distinct one-step successors of a reachable
grouping have a common one-step successor.}
Let $\{P,Q\}$ and $\{R,S\}$ be distinct mutual pairs of a reachable grouping
$\mathcal P$. Suppose first that they are disjoint. Merging one leaves the
other mutual, since its two blocks remain blocks and two-sided attachment does
not depend on the grouping. Merging both, in either order, yields the grouping
in which $P$ and $Q$ are replaced by $P\cup Q$, and $R$ and $S$ by $R\cup S$.
Suppose next that they share a block, say $Q=R$ with $P\neq S$. After merging
$\{P,Q\}$ we have $\bar g(P\cup Q,S)\geq\bar g(Q,S)=1$ and
$\bar g(S,P\cup Q)\geq\bar g(S,Q)=1$, so $\{P\cup Q,S\}$ is a mutual pair.
Merging it yields the grouping in which $P$, $Q$, and $S$ are replaced by
$P\cup Q\cup S$. Merging $\{Q,S\}$ first and then $\{Q\cup S,P\}$ yields the
same grouping.

\smallskip
\noindent\emph{Step 2: all maximal sequences of mergers from a reachable
grouping terminate at the same endpoint.}
We induct on the number of blocks; the case $\mathcal P=\mathcal I$ is the
proposition. If $\mathcal P$ has no mutual pair, in particular if it has one
block, every maximal sequence from $\mathcal P$ is empty. Otherwise, consider
two maximal sequences from $\mathcal P$, starting with
$\mathcal P\to\mathcal P_1$ and $\mathcal P\to\mathcal P_2$. Both
$\mathcal P_1$ and $\mathcal P_2$ have one block fewer than $\mathcal P$. By
the induction hypothesis, all maximal sequences from $\mathcal P_1$ terminate
at a common endpoint $\mathcal N_1$, and all maximal sequences from
$\mathcal P_2$ terminate at a common endpoint $\mathcal N_2$. The two given
sequences therefore end at $\mathcal N_1$ and $\mathcal N_2$. If
$\mathcal P_1=\mathcal P_2$, then $\mathcal N_1=\mathcal N_2$. Otherwise the
merged pairs are distinct, so Step~1 gives a common one-step successor
$\mathcal P_{12}$. Fix a maximal sequence from $\mathcal P_{12}$, with
endpoint $\mathcal N$. Prepending the step
$\mathcal P_1\to\mathcal P_{12}$ gives a maximal sequence from
$\mathcal P_1$, and prepending $\mathcal P_2\to\mathcal P_{12}$ gives one from
$\mathcal P_2$. Hence $\mathcal N_1=\mathcal N=\mathcal N_2$.
\end{proof}

We denote the terminal minor of $(I,g)$ by $(\mathcal Q^*,g_{\mathcal Q^*})$.
The next lemma relates its blocks to the induced subnetworks of
Section~\ref{sec:local}.

\begin{lemma}\label{lem:blocks-pseudo-reciprocal}
\textnormal{(i)} For every $P\in\mathcal Q^*$, the induced subnetwork
$(P,g_P)$ is strongly connected and pseudo-reciprocal.
\textnormal{(ii)} If $(J,g_J)$ is weakly connected and pseudo-reciprocal,
then $J$ is contained in a single block of $\mathcal Q^*$.
\end{lemma}

\begin{proof}
(i) Fix a maximal sequence of mergers
$\mathcal I=\mathcal P_0\to\cdots\to\mathcal P_m=\mathcal Q^*$ and a block
$P\in\mathcal Q^*$. Since groupings coarsen along the sequence, every block
of every $\mathcal P_t$ is contained in some block of $\mathcal Q^*$. Hence
the blocks of $\mathcal P_t$ contained in $P$ form a partition of $P$, and
the two blocks merged at any step lie in the same block of $\mathcal Q^*$.
Call a step a \emph{$P$-step} if it merges two blocks contained in $P$; all
other steps leave the blocks contained in $P$ unchanged.

We claim that the $P$-steps form a merger sequence of $(P,g_P)$ from the
singleton grouping of $P$ to $\{P\}$. By
Lemma~\ref{lem:closed-form-mergers}(iii), a $P$-step merges blocks
$A,B\subseteq P$ with $A\rightleftarrows_g B$. The witnessing links have both
endpoints in $P$, so $A$ and $B$ are two-sidedly attached in $(P,g_P)$ as
well. Applying Lemma~\ref{lem:closed-form-mergers}(iii) to $(P,g_P)$, the
pair $\{A,B\}$ is a mutual pair of the grouping of $P$ reached by the
preceding $P$-steps. By induction along the sequence, the $P$-steps are valid
mergers in $(P,g_P)$. They terminate at $\{P\}$, because $P$ is a single
block of $\mathcal Q^*$. Lemma~\ref{lem:pseudo-one-block}(i) then shows that
$(P,g_P)$ is pseudo-reciprocal. Lemma~\ref{lem:block-connectivity}, applied
to the terminal minor, gives a directed path within $P$ from each agent of
$P$ to every other. Hence $(P,g_P)$ is strongly connected.

(ii) If $|J|=1$, the claim is trivial. Otherwise, by
Lemma~\ref{lem:pseudo-one-block}(ii) applied to $(J,g_J)$, some merger
sequence carries the singleton grouping of $J$ to $\{J\}$. Perform the same
steps in $(\mathcal I,g_{\mathcal I})$, leaving each singleton $\{i\}$ with
$i\notin J$ untouched. Each pair merged is two-sidedly attached in $(J,g_J)$,
hence in $(I,g)$, and is therefore a mutual pair of the current grouping by
Lemma~\ref{lem:closed-form-mergers}(iii). The resulting sequence reaches a
grouping in which $J$ is a block. Extend it to a maximal sequence, which
terminates at $\mathcal Q^*$ by Proposition~\ref{prop:confluence}. Since
groupings only coarsen along the sequence, $J$ is contained in a block of
$\mathcal Q^*$.
\end{proof}

\subsection{Fixed points and agreement}

\begin{lemma}\label{lem:fixed-point-characterization}
For a decision environment $\mathscr D=(\Omega,D,f)$ and a profile $\Pi\in\boldsymbol\Pi^n$,
\[
    \Pi\in F_{\mathscr D}(I,g)
    \quad\Longleftrightarrow\quad
    W(\Pi_j)\leq \Pi_i\text{ for every }i\in I\text{ and every }j\in S(i).
\]
\end{lemma}

\begin{proof}
For each $i\in I$,
\[
    h_i(\Pi)=\Pi_i
    \quad\Longleftrightarrow\quad
    \Pi_i\vee\bigvee_{j\in S(i)}W(\Pi_j)=\Pi_i
    \quad\Longleftrightarrow\quad
    W(\Pi_j)\leq\Pi_i\quad \forall j\in S(i),
\]
where the condition is vacuous when $S(i)=\emptyset$, since the empty join is
$\{\Omega\}$ and hence $\Pi_i\vee\{\Omega\}=\Pi_i$.
\end{proof}

The next lemma is essentially the classical agreement theorem of \citet{cave} and
\citet{bacharach}.

\begin{lemma}\label{lem:cave-bacharach}
Let $\mathscr D=(\Omega,D,f)$ be a decision environment. If two partitions $\Pi,\Pi^{\prime}$ satisfy $W(\Pi)\leq\Pi^{\prime}$ and $W(\Pi^{\prime})\leq\Pi$, then $f\circ\Pi=f\circ\Pi^{\prime}$.
\end{lemma}

\begin{proof}
Fix $\omega\in\Omega$ and let $C:=(\Pi\wedge\Pi^{\prime})(\omega)$. The map
$f\circ\Pi$ is constant on every $\Pi$-cell by definition. It is also
constant on every $\Pi^{\prime}$-cell, because $W(\Pi)\leq\Pi^{\prime}$. The
cell $C$ is the set of states reachable from $\omega$ by finitely many steps,
each staying within a $\Pi$-cell or within a $\Pi^{\prime}$-cell. Hence
$f\circ\Pi$ is constant on $C$. Applying STP to the $\Pi$-cells partitioning
$C$ gives $f(C)=f(\Pi(\omega))$, and symmetrically
$f(C)=f(\Pi^{\prime}(\omega))$. Thus
$f(\Pi(\omega))=f(\Pi^{\prime}(\omega))$ for every $\omega$.
\end{proof}

\begin{lemma}\label{lem:block-consensus}
Let $\mathscr D=(\Omega,D,f)$ be a decision environment, and let $(\mathcal P,g_{\mathcal P})\preceq(\mathcal I,g_{\mathcal I})$. If
$\Pi\in F_{\mathscr D}(I,g)$, then for every block $P\in\mathcal P$, we have $f_a^\Pi=f_b^\Pi$ for any $a,b\in P$.
\end{lemma}

\begin{proof}
We induct on the length of the merger sequence producing $\mathcal P$. Singleton
blocks satisfy the claim. Suppose the claim holds at $\mathcal P_t$ and
$\mathcal P_{t+1}$ is obtained by merging a mutual pair $P,Q\in\mathcal P_t$.
Only the new block $P\cup Q$ needs to be checked.

By the induction hypothesis, all agents in $P$ share a decision rule $f_P$,
and all agents in $Q$ share a decision rule $f_Q$. Since $P$ and $Q$ are
mutually linked in the grouped network,
Lemma~\ref{lem:closed-form-mergers}(iii) yields agents $a_1\in P,\ b_1\in Q$
and $a_2\in Q,\ b_2\in P$ with $g(a_1,b_1)=1$ and $g(a_2,b_2)=1$.
Lemma~\ref{lem:fixed-point-characterization} gives
$W(\Pi_{a_1})\leq\Pi_{b_1}$ and $W(\Pi_{a_2})\leq\Pi_{b_2}$.

Because $a_1,b_2\in P$ and $a_2,b_1\in Q$, the induction hypothesis gives
$f_{a_1}^\Pi=f_{b_2}^\Pi$ and $f_{a_2}^\Pi=f_{b_1}^\Pi$. By definition,
$W(\Pi_i)$ depends on $\Pi_i$ only through $f_i^\Pi$. Hence
$W(\Pi_{a_1})=W(\Pi_{b_2})$ and $W(\Pi_{a_2})=W(\Pi_{b_1})$, so that
$W(\Pi_{b_2})\leq\Pi_{b_1}$ and $W(\Pi_{b_1})\leq\Pi_{b_2}$. By
Lemma~\ref{lem:cave-bacharach}, $f_{b_1}^\Pi=f_{b_2}^\Pi$. Therefore
$f_P=f_Q$, and all agents in $P\cup Q$ share the same decision rule.
\end{proof}

\subsection{The separating environment}\label{sec:proofs-separation}

\begin{lemma}\label{lem:separation}
Let $(\mathcal Q,g_{\mathcal Q})$ be a grouped network with $|\mathcal Q|\geq 2$
and $\mathscr M(\mathcal Q,g_{\mathcal Q})=\emptyset$. Then there exist a
decision environment $\mathscr D=(\Omega,D,f)$ and partitions
$(\Pi_R)_{R\in\mathcal Q}$ of $\Omega$, one for each block, with the following
two properties:
\begin{equation}\label{eq:block-fixed-point}
    W(\Pi_S)\leq\Pi_R
    \qquad\text{whenever }g_{\mathcal Q}(S,R)=1;
\end{equation}
and $f\circ\Pi_R\neq f\circ\Pi_S$ for all distinct blocks
$R,S\in\mathcal Q$.
\end{lemma}

In the construction below, each block $R$ carries a binary attribute. The
partition $\Pi_R$ reveals the attributes of $R$ and of the blocks sending to
$R$, while the decision made under $\Pi_R$ discloses the attribute of $R$
alone. The absence of mutual pairs is precisely what makes this disclosure
rule well defined (Step~2).

\begin{proof}
\emph{Step 1: construction.}
For each block $R\in\mathcal Q$, define its closed in-neighborhood in the
grouped network by $N_R:=\{R\}\cup\{S\in\mathcal Q:g_{\mathcal Q}(S,R)=1\}$.
Let $\Omega:=\{0,1\}^{\mathcal Q}$. For $\omega\in\Omega$, write $\omega_R$
for the coordinate indexed by $R$. Define $\Pi_R$ to be the partition
generated by the coordinates in $N_R$, so that
$\omega\equiv_{\Pi_R}\omega^{\prime}$ if and only if
$\omega_T=\omega^{\prime}_T$ for every $T\in N_R$. For $b\in\{0,1\}$, let
$E_{R,b}:=\{\omega\in\Omega:\omega_R=b\}$. Since $R\in N_R$, every
$\Pi_R$-cell is contained in $E_{R,0}$ or in $E_{R,1}$; a $\Pi_R$-cell
contained in $E_{R,0}$ is called a \emph{0-cell of $R$}.

Let $\mathcal A$ be the collection of all non-empty events that can be
written as a disjoint union of 0-cells, where the 0-cells may belong to
different blocks. For each $R\in\mathcal Q$, let $\mathcal U_R$ be the
collection of all non-empty events $E\subseteq\Omega$ such that $E$ is a
union of $\Pi_R$-cells and $E\subseteq E_{R,1}$. Let $a$ and
$(b_R)_{R\in\mathcal Q}$ be distinct symbols outside $2^{\Omega}$, and set
$D:=\{a\}\cup\{b_R:R\in\mathcal Q\}\cup 2^{\Omega}$. Define
$f:2^{\Omega}\to D$ by
\begin{equation*}
    f(E):=
    \begin{cases}
    a, & E\in\mathcal A,\\
    b_R, & E\in\mathcal U_R,\\
    E, & \text{otherwise.}
    \end{cases}
\end{equation*}

\smallskip
\noindent\emph{Step 2: $f$ is well defined.}
First, $\mathcal U_R\cap\mathcal U_S=\emptyset$ for distinct
$R,S\in\mathcal Q$. Suppose $E\in\mathcal U_R\cap\mathcal U_S$. The event $E$
is non-empty and is a union of $\Pi_R$-cells. Hence, if $S\notin N_R$,
flipping only the $S$-coordinate of any state in $E$ would keep the state in
$E$, contradicting $E\subseteq E_{S,1}$. Thus $S\in N_R$, and symmetrically
$R\in N_S$. These inclusions imply
$g_{\mathcal Q}(S,R)=g_{\mathcal Q}(R,S)=1$, which contradicts
$\mathscr M(\mathcal Q,g_{\mathcal Q})=\emptyset$.

Second, $\mathcal A\cap\mathcal U_R=\emptyset$ for every $R\in\mathcal Q$.
Suppose $U\in\mathcal A\cap\mathcal U_R$, and write
$U=\bigsqcup_{k=1}^{K}C_k$, where each $C_k$ is a 0-cell of some block
$S_k$. Fix $k$. The cell $C_k$ fixes exactly the coordinates in $N_{S_k}$,
and $C_k\subseteq U\subseteq E_{R,1}$. If $R\notin N_{S_k}$, flipping the
$R$-coordinate of a state in $C_k$ would keep it in $C_k$, contradicting
$C_k\subseteq E_{R,1}$. Hence $R\in N_{S_k}$. Moreover $S_k\neq R$, since
$\omega_{S_k}=0$ and $\omega_R=1$ on $C_k$. Thus
$g_{\mathcal Q}(R,S_k)=1$, and
$\mathscr M(\mathcal Q,g_{\mathcal Q})=\emptyset$ forces
$g_{\mathcal Q}(S_k,R)=0$, that is, $S_k\notin N_R$. Now take any
$\omega\in U$, and let $\omega^*$ be the state obtained from $\omega$ by
setting the coordinates $S_1,\ldots,S_K$ to $1$. The state $\omega^*$ agrees
with $\omega$ on every coordinate in $N_R$, and $U$ is a union of
$\Pi_R$-cells, so $\omega^*\in U$. But $\omega^*_{S_k}=1$ implies
$\omega^*\notin C_k$ for every $k$, contradicting
$U=\bigsqcup_{k=1}^{K}C_k$.

\smallskip
\noindent\emph{Step 3: $f$ satisfies STP.}
Let $\mathcal E$ be a non-empty family of pairwise disjoint non-empty events
such that $f(E)=d$ for every $E\in\mathcal E$. If $d=a$, then each member of
$\mathcal E$ is a disjoint union of 0-cells. Since the members are pairwise
disjoint, $\bigcup_{E\in\mathcal E}E$ is again a non-empty disjoint union of
0-cells, so it belongs to $\mathcal A$ and receives value $a$. If $d=b_R$,
then each member is a non-empty union of $\Pi_R$-cells contained in
$E_{R,1}$, and so is their union, which therefore belongs to $\mathcal U_R$
and receives value $b_R$. If $d\in 2^{\Omega}$, then every member of
$\mathcal E$ falls under the third case and equals $d$. Pairwise disjointness
and non-emptiness then give $\mathcal E=\{d\}$, so
$\bigcup_{E\in\mathcal E}E=d$ and $f(d)=d$.

\smallskip
\noindent\emph{Step 4: working partitions and the fixed-point condition.}
Fix $R\in\mathcal Q$. Every $\Pi_R$-cell contained in $E_{R,0}$ is a 0-cell
of $R$ and belongs to $\mathcal A$, receiving value $a$. Every $\Pi_R$-cell
contained in $E_{R,1}$ belongs to $\mathcal U_R$, receiving value $b_R$.
Since $a\neq b_R$, it follows that $W(\Pi_R)=\{E_{R,0},E_{R,1}\}$. Now
suppose $g_{\mathcal Q}(S,R)=1$. Then $S\in N_R$, so every $\Pi_R$-cell fixes
the $S$-coordinate. Thus $\Pi_R$ refines $\{E_{S,0},E_{S,1}\}=W(\Pi_S)$, or
equivalently $W(\Pi_S)\leq\Pi_R$. This proves
\eqref{eq:block-fixed-point}.

\smallskip
\noindent\emph{Step 5: separation.}
Let $R,S\in\mathcal Q$ be distinct blocks, which exist since
$|\mathcal Q|\geq 2$, and choose $\omega\in\Omega$ with $\omega_R=1$.
Then $f(\Pi_R(\omega))=b_R$, whereas $f(\Pi_S(\omega))\in\{a,b_S\}$. Since
$b_R\notin\{a,b_S\}$, we conclude $f\circ\Pi_R\neq f\circ\Pi_S$.
\end{proof}

The separating environment lifts from the terminal minor to the original
network as follows.

\begin{corollary}\label{cor:lift}
If $|\mathcal Q^*|\geq 2$, then there exist a finite decision environment
$\mathscr D=(\Omega,D,f)$ and a profile $\Pi\in F_{\mathscr D}(I,g)$ such
that $f_i^\Pi\neq f_j^\Pi$ whenever $i$ and $j$ belong to distinct blocks
of $\mathcal Q^*$.
\end{corollary}

\begin{proof}
By the definition of the terminal minor,
$\mathscr M(\mathcal Q^*,g_{\mathcal Q^*})=\emptyset$, so
Lemma~\ref{lem:separation} applies to $(\mathcal Q^*,g_{\mathcal Q^*})$ and
yields a decision environment $\mathscr D$, with finite state space
$\{0,1\}^{\mathcal Q^*}$, and partitions
$(\Pi_R)_{R\in\mathcal Q^*}$. For each $i\in I$, let $R_i$ be the block of
$\mathcal Q^*$ containing $i$, and set $\Pi_i:=\Pi_{R_i}$.

We verify that $\Pi\in F_{\mathscr D}(I,g)$. Let $u\to v$ be a link of
$(I,g)$. If $R_u=R_v$, then $W(\Pi_u)\leq\Pi_u=\Pi_v$, since every working
partition is coarser than the partition inducing it. If $R_u\neq R_v$, then
$g_{\mathcal Q^*}(R_u,R_v)=\bar g(R_u,R_v)\geq g(u,v)=1$ by
Lemma~\ref{lem:closed-form-mergers}(iii), and \eqref{eq:block-fixed-point}
gives $W(\Pi_u)=W(\Pi_{R_u})\leq\Pi_{R_v}=\Pi_v$. Hence $W(\Pi_u)\leq\Pi_v$
for every link, and Lemma~\ref{lem:fixed-point-characterization} gives
$\Pi\in F_{\mathscr D}(I,g)$.

Finally, if $R_i\neq R_j$, then
$f_i^\Pi=f\circ\Pi_{R_i}\neq f\circ\Pi_{R_j}=f_j^\Pi$ by
Lemma~\ref{lem:separation}.
\end{proof}

\subsection{Proofs of the main results}

\begin{proof}[Proof of Theorem~\ref{thm:pseudo-reciprocity}]
\emph{Step 1: pseudo-reciprocity implies that the network ensures consensus.}
Assume that $(I,g)$ is pseudo-reciprocal. Since $(I,g)$ is weakly connected
(Assumption~\ref{ass:weak-connectedness}),
Lemma~\ref{lem:pseudo-one-block}(ii) shows that the one-block grouping $\{I\}$
is obtainable from $(\mathcal I,g_{\mathcal I})$ by a finite sequence of
mergers.

Fix any decision environment $\mathscr D=(\Omega,D,f)$ and any
$\Pi\in F_{\mathscr D}(I,g)$. Applying
Lemma~\ref{lem:block-consensus} to the one-block grouping gives
\[
    f_i^\Pi=f_j^\Pi
    \qquad\forall i,j\in I.
\]
Thus $\Pi\in A_{\mathscr D}$. Since $\mathscr D$ and $\Pi$ were arbitrary,
$(I,g)$ ensures consensus.

\smallskip
\noindent\emph{Step 2: failure of pseudo-reciprocity implies that the network does not ensure consensus.}
Assume that $(I,g)$ is not pseudo-reciprocal. Then $|\mathcal Q^*|\geq 2$:
otherwise $\mathcal Q^*=\{I\}$, and the one-block grouping would be a
reciprocal mutual minor of $(\mathcal I,g_{\mathcal I})$, contradicting the
assumption. By Corollary~\ref{cor:lift}, there exist a decision environment
$\mathscr D$ and a profile $\Pi\in F_{\mathscr D}(I,g)$ with
$f_i^\Pi\neq f_j^\Pi$ for any agents $i$ and $j$ in distinct blocks of
$\mathcal Q^*$; such agents exist because $|\mathcal Q^*|\geq 2$. Thus
$\Pi\notin A_{\mathscr D}$, and $(I,g)$ fails to ensure consensus.
\end{proof}

\begin{proof}[Proof of Corollary~\ref{cor:greedy}]
\textnormal{(i)} $\Rightarrow$ \textnormal{(iii)}: by
Assumption~\ref{ass:weak-connectedness} and
Lemma~\ref{lem:pseudo-one-block}(ii), some sequence of mergers reaches
$\{I\}$, and it is maximal because $\{I\}$ has no mutual pair. By
Proposition~\ref{prop:confluence}, every maximal sequence then terminates at
$\{I\}$. \textnormal{(iii)} implies \textnormal{(ii)} because a maximal
sequence exists, and \textnormal{(ii)} implies \textnormal{(i)} because the
one-block grouping is reciprocal.
\end{proof}

\begin{proof}[Proof of Corollary~\ref{cor:single-core-characterization}]
By Theorem~\ref{thm:pseudo-reciprocity}, \textnormal{(i)} and \textnormal{(ii)}
are equivalent. For the equivalence of \textnormal{(ii)} and
\textnormal{(iii)}, recall from Lemma~\ref{lem:closed-form-mergers}(iii)
that, in any grouped network obtained from $(\mathcal I,g_{\mathcal I})$ by
mergers, two blocks form a mutual pair if and only if they are two-sidedly
attached.

\smallskip
\noindent\emph{\textnormal{(iii)} $\Rightarrow$ \textnormal{(ii)}.}
Let $v_1,\ldots,v_n$ be a single-core order. Since
$\{v_1,v_2\}\in\mathscr M(I,g)$, the singletons $\{v_1\}$ and $\{v_2\}$ are a
mutual pair and merge into the block $M_2$. Suppose the core $M_k$ has been
formed as a single block, with $2\le k<n$. By \eqref{eq:single-core-order},
$M_k\rightleftarrows_g\{v_{k+1}\}$, so $\{M_k,\{v_{k+1}\}\}$ is a mutual pair
and merges into $M_{k+1}$. Iterating this step reaches the one-block grouping
$\{I\}$, so $(I,g)$ is pseudo-reciprocal by
Lemma~\ref{lem:pseudo-one-block}(i). With
$\mathscr M(I,g)=\{\{v_1,v_2\}\}$ from \eqref{eq:single-core-order},
\textnormal{(ii)} follows.

\smallskip
\noindent\emph{\textnormal{(ii)} $\Rightarrow$ \textnormal{(iii)}.}
Suppose $(I,g)$ has a unique direct mutual pair $\{p,q\}$ and is
pseudo-reciprocal. By Assumption~\ref{ass:weak-connectedness} and
Lemma~\ref{lem:pseudo-one-block}(ii), some sequence of mergers carries
$(\mathcal I,g_{\mathcal I})$ to $\{I\}$; fix one. Each merger acts on a
mutual pair of the current grouping. The first acts on singletons, hence on a
direct mutual pair, so by uniqueness it merges $\{p\}$ and $\{q\}$ into
$M_2:=\{p,q\}$. No two singletons can be merged thereafter, as that would
exhibit a second direct mutual pair. Hence exactly one block is non-singleton
from then on, and each merger attaches a singleton $\{v_{k+1}\}$ to the
current core $M_k$, which requires $M_k\rightleftarrows_g\{v_{k+1}\}$. Since
the sequence reaches $\{I\}$, this orders all agents. Setting $v_1:=p$ and
$v_2:=q$ yields \eqref{eq:single-core-order}, so $(I,g)$ admits a single-core
order.
\end{proof}

\begin{proof}[Proof of Proposition~\ref{prop:link-count}]
Suppose $(I,g)$ ensures consensus. By Theorem~\ref{thm:pseudo-reciprocity},
$(I,g)$ is pseudo-reciprocal. Assumption~\ref{ass:weak-connectedness} and
Lemma~\ref{lem:pseudo-one-block}(ii) then give a merger sequence
$\mathcal I=\mathcal P_0\to\cdots\to\mathcal P_{n-1}=\{I\}$. Fix
$t\in\{1,\ldots,n-1\}$ and consider the $t$-th step, which merges two blocks
$P,Q\in\mathcal P_{t-1}$. By Lemma~\ref{lem:closed-form-mergers}(iii), there
are links $(a,b)$ and $(c,d)$ with $g(a,b)=g(c,d)=1$, $a,d\in P$, and
$b,c\in Q$; they are distinct, since $a\in P$ and $c\in Q$. Call a link
\emph{internal} at stage $t$ if its endpoints lie in the same block of
$\mathcal P_t$. The two links above are not internal at stage $t-1$ but are
internal at stage $t$. Since the groupings coarsen along the sequence, a link
internal at some stage remains internal at all later stages. Hence the pairs
of links associated with distinct steps are disjoint, and the $n-1$ steps
yield $2(n-1)$ distinct links.

For the second statement, fix $n\geq 2$ and let $I=\{v_1,\ldots,v_n\}$ with
$g(v_1,v_2)=g(v_2,v_1)=1$, $g(v_1,v_{k+1})=g(v_{k+1},v_k)=1$ for each
$k=2,\ldots,n-1$, and no other links. Let $m\geq 3$. The only link out of
$v_m$ ends at $v_{m-1}$, and the links into $v_m$ originate at $v_1$ or,
when $m<n$, at $v_{m+1}$. Since $v_{m-1}\notin\{v_1,v_{m+1}\}$, no direct
mutual pair involves $v_m$. Hence $\mathscr M(I,g)=\{\{v_1,v_2\}\}$. Moreover
$M_k\rightleftarrows_g\{v_{k+1}\}$ for every $k=2,\ldots,n-1$, so
$v_1,\ldots,v_n$ is a single-core order, and the network ensures consensus by
Corollary~\ref{cor:single-core-characterization}. It has exactly $2(n-1)$
links. Figure~\ref{fig:reciprocal-core} depicts another network for which
equality holds.
\end{proof}

\begin{proof}[Proof of Theorem~\ref{thm:decomposition}]
\textnormal{(i)} $\Rightarrow$ \textnormal{(ii)}: suppose \textnormal{(ii)}
fails. If $i$ and $j$ were in the same block of $\mathcal Q^*$, that block
would witness \textnormal{(ii)} by
Lemma~\ref{lem:blocks-pseudo-reciprocal}(i). Hence $i$ and $j$ lie in
distinct blocks of $\mathcal Q^*$, and in particular $|\mathcal Q^*|\geq 2$.
Corollary~\ref{cor:lift} then yields a decision environment $\mathscr D$
and a profile $\Pi\in F_{\mathscr D}(I,g)$ with $f_i^\Pi\neq f_j^\Pi$, so
\textnormal{(i)} fails.

\textnormal{(ii)} $\Rightarrow$ \textnormal{(iii)}: by
Lemma~\ref{lem:blocks-pseudo-reciprocal}(ii), $J$ is contained in a single
block of $\mathcal Q^*$, which then contains both $i$ and $j$.

\textnormal{(iii)} $\Rightarrow$ \textnormal{(i)}: since
$(\mathcal Q^*,g_{\mathcal Q^*})\preceq(\mathcal I,g_{\mathcal I})$,
Lemma~\ref{lem:block-consensus} gives $f_i^\Pi=f_j^\Pi$ for every decision
environment $\mathscr D$ and every $\Pi\in F_{\mathscr D}(I,g)$.

Finally, every block of $\mathcal Q^*$ induces a weakly connected,
pseudo-reciprocal subnetwork by
Lemma~\ref{lem:blocks-pseudo-reciprocal}(i), and every set inducing one is
contained in a block by Lemma~\ref{lem:blocks-pseudo-reciprocal}(ii). Since
distinct blocks are disjoint, a set inducing such a subnetwork and
containing a block coincides with that block, so each block is maximal.
Conversely, a maximal such set is contained in a block, which induces such
a subnetwork itself, and therefore coincides with it.
\end{proof}

\section*{Acknowledgements}
I am grateful to Nozomu Muto, Norio Takeoka, and Hendrik Rommeswinkel for
their valuable comments and suggestions.

\end{document}